\documentclass[showpacs,amsmath,amssymb,onecolumn,superscriptaddress,notitlepage,preprintnumbers,pra]{revtex4-1}
\usepackage[dvips]{graphicx}
\usepackage{subfigure}
\usepackage{amsmath,amssymb,amsthm,mathrsfs,amsfonts,dsfont}
\usepackage{dcolumn}
\usepackage{amsmath,amsthm,amssymb,amsfonts,latexsym}
\usepackage{bm}
\usepackage{braket}
\usepackage{physics}
\usepackage{mathtools}
\usepackage{diagbox}
\usepackage[utf8]{inputenc}
\usepackage[english]{babel}
\usepackage{scalerel}[2014/03/10]
\usepackage{stackengine}
\usepackage[noend,noline,ruled,linesnumbered]{algorithm2e}
\usepackage[noend]{algpseudocode} 
\usepackage{makecell}
\usepackage{lipsum}
\usepackage{hyperref}
\usepackage{multirow}
\usepackage{makecell}
\usepackage{float}
\restylefloat{table}
\usepackage{placeins}

\renewcommand\title[1]{\bf \hskip2.25pc \parbox{.8\textwidth}{ \noindent%
   \LARGE \bf \begin{center} #1 \end{center} \rm } \vskip.1in \rm\normalsize }

\hypersetup{colorlinks=true, linkcolor=blue, citecolor=blue, urlcolor=black }
\makeatletter

\newtheorem{lemma}{\textbf{Lemma}}

\newtheorem{theorem}{\textbf{Theorem}}
\newtheorem{definition}{\textbf{Definition}}
\newtheorem{fact}{\textbf{Fact}}

\newcommand{\payoff}{\mathrm{payoff}}

\begin{document}
%\title{Quantum-enhanced VaR and CVaR analysis on option portfolios}

\title{
\Large Quantum Computing for Option Portfolio Analysis}
%The power of quantum computation in VaR and CVaR analysis of option portfolios}

\author{Yusen Wu}
\affiliation{
Department of Physics, The University of Western Australia, Perth, WA 6009, Australia
}

\author{Jingbo B. Wang}
\email{jingbo.wang@uwa.edu.au}
\affiliation{
Department of Physics, The University of Western Australia, Perth, WA 6009, Australia
}

\author{Yuying Li}
\affiliation{
Cheriton School of Computer Science, University of Waterloo, Waterloo, Canada}

\date{\today}

\begin{abstract}
In this paper, we introduce an efficient and end-to-end quantum algorithm tailored for computing the Value-at-Risk (VaR) and conditional Value-at-Risk (CVar) for a portfolio of European options. 
Our focus is on leveraging quantum computation to overcome the challenges posed by high dimensionality in VaR and CVaR estimation. 
While our innovative quantum algorithm is designed primarily for estimating portfolio VaR and CVaR for European options, we also investigate the feasibility of applying a similar quantum approach to price American options. Our analysis reveals a quantum 'no-go' theorem within the current algorithm, highlighting its limitation in pricing American options.  Our results indicate the necessity of investigating alternative strategies to resolve the complementarity challenge in pricing American options in future research.

\end{abstract}

\maketitle

\section{Introduction}
Financial institutions such as banks and insurance companies often manage extensive portfolios of financial derivatives that require regular risk assessment.
The Basel Accords have embraced a measure of market tail risk in global banking regulation, including using traditional Value-at-Risk (VaR) and Conditional Value-at-Risk (CVaR).
For a given time horizon $[0,\bar{t}]$, where $\bar{t}$ could be a day, a week or longer from now and before the expiry of any option in the portfolio, VaR is the maximum potential loss within a specified confidence level (e.g. 95\%), and CVaR represents the expected value of losses that exceed the VaR threshold. 
Regularly calculating such risk values is crucial and often required by regulatory bodies and institutional risk management.
As an example, an insurance company needs to compute VaR/CVaR of a large portfolio of hundreds of thousands of (potentially very complex) option contracts, as discussed in \cite{Brodie11,Gan15,Gan20,Lucio22}.
Estimating the risk of a large portfolio of options is a computationally challenging task faced by financial institutions and insurance companies, because computing VaR and CVaR for a large portfolio of complex options
is a computationally intensive task in classical computing.  

The aforementioned challenge can be appreciated by considering the following. Computing VaR of a portfolio of financial options requires calculating the time $\bar{t}$ value of each option for each possible underlying price realization at $\bar{t}$.  Following classical no-arbitrage option pricing theory, the fair value of an option is derived from the value of the underlying asset $S$ and contract specification parameters, such as strikes and time to maturity. At any given future time  $t>0$,  the fair value of the option is some nonlinear function $V(S,t)$ of the underlying price $S$ and time $t$. 
%Since the underlying price $S_t$ is random, 
Since the underlying price $S_t$ follows a stochastic process, it exhibits a  behavior characterized by random fluctuations over time.
Consequently, the distribution of the value of a portfolio of options is determined by distributions of option values in the portfolio, which are given by distributions of underlying assets and corresponding option value functions.  
%\iffalse
%An option is a financial contract that gives the holder the right, but not the obligation, to buy or sell an underlying asset at a predetermined price within a specified time frame. This right to buy or sell is given at a predetermined price, known as the strike or exercise price. Options are a type of derivative, meaning their worth is derived from the value of an underlying asset, such as stocks, commodities, currencies, or other financial instruments.
%\fi
Calculating VaR/CVaR of a portfolio of options at a future time horizon $\bar{t}>0$ requires the availability of joint distributions of all underlying assets and all option value functions. In general,  VaR analysis consists of two significant components:  (1) simulation of the time $\bar{t}$ price joint distributions of all underlying assets of options in the portfolio and (2) computation of fair option values for all underlying asset price realizations at time $\bar{t}$.

% classical hardness of VaR analysis
For simplicity, in this paper,  we assume that the underlying of an option is a single asset but underlying asset prices across different options can be correlated. In order to generate accurate option values, it is well known that a constant volatility Black-Scholes (BS) model is inadequate~\cite{hull2012options}, as evidenced by the implied volatility smile \cite{Dup94,Der94,Dum98,CLV98}. The local volatility function model is one of the simplest generalizations of the BS model, which has been widely studied in the literature, e.g., \cite{CLV98,Andersen00,He06}, in addition to being frequently adopted in industry practice. More sophisticated models used in option pricing include generalizations of the diffusion BS model augmented with stochastic jumps \cite{Merton76}, stochastic volatility \cite{Hes93}, and more recently, rough volatility \cite{Gat18}. The fair value of the option is generally characterized by the unique solution to a partial differential equation, which can suffer the curse of dimensionality and becomes computationally infeasible on a classical computer when when there are more than $3$ assets. 
%the dimension becomes greater than $3$ assets. 

% quantum algorithm for VaR
Recent advances in quantum computational processors have demonstrated significant quantum computational advantages in problems such as random quantum state sampling~\cite{boixo2018characterizing,arute2019quantum,zhong2020quantum} and density matrix/quantum channel property learning~\cite{huang2022quantum, wu2023quantum}. Given the significance of these outcomes, quantum computers are expected to be capable of accelerating finance models, such as option value modelling~\cite{miyamoto2021pricing, rebentrost2018quantum, gonzalez2021pricing, stamatopoulos2020option}, portfolio management~\cite{woerner2019quantum, kerenidis2020quantum, qu2024experimental,slate2021quantum} and forecasting anomaly detection recommendations~\cite{low2014quantum, borujeni2021quantum, moreira2016quantum}. In the context of option value modelling, efficient quantum algorithms have been proposed for solving European options and Asian options~\cite{miyamoto2021pricing, rebentrost2018quantum, gonzalez2021pricing}, which approximately transform the differential equation into a system of linear equations on discretized grid points. Although under various assumptions on the conditioning of the linear systems and data-access model, the quantum linear system can be solved with an exponential reduction in the dependence on the dimension~\cite{harrow2009quantum}, a sampling cost needs to be paid to read out the solution from the quantum state, making this technique lose its original quantum advantage. As a result, a rigorous quantum advantage is still unclear for pricing risk modelling problems.

In this paper, we investigate the power of quantum computation in computing the VaR and CVaR for option portfolios. Specifically, when the fair option value function $V(S,t)$ is determined by a PDE method for European options, we design an efficient quantum approach to estimate the VaR/CVaR value, where the involved fundamental steps are provided in Sec.~\ref{Sec:outline}. The proposed quantum algorithm requires a $\tilde{\mathcal{O}}(\max\{\bar{t},T-\bar{t}\}\epsilon_d^{-1/2}\epsilon^{-4})$-depth quantum circuit, where $\epsilon$ represents the additive error to VaR/CVaR estimation and $\epsilon_d$ represents the error induced by the discretization approach. Although the computation of VaR/CVaR value requires the joint distribution of all underlying assets and option value functions, the proposed quantum algorithm does not require expensive classical post-processing (large sampling cost to read out the solution) and thus maintains the potential quantum speed-up.  The main goal of this paper is to present a quantum implementation for 
a hybrid PDE-MC approach to compute VaR/CVaR of a large portfolio of options, e.g., we are interested in computing, e.g., 5\% VaR/CVaR at time $\bar{t}>0$.

\section{Theoretical Background}

\subsection{Black Scholes Model}
To illustrate option pricing on a single underlying asset, we assume a generalized Black Scholes local volatility function model for an underlying asset price evolution,
\begin{equation}\label{eq:LVFReal}
\frac{dS(t)}{S(t)} = \mu {\rm \bm d}t + \sigma(S(t),t) {\rm \bm d}Z_t ,
\end{equation}
where $Z_t$ represents a standard Brownian motion, $\mu$ a drift, and $\sigma(S(t),t)$ a local volatility function.
This stochastic differential equation (SDE)
does not have a closed-form solution in general,  and consequently there can be no analytic formula for the distribution function for $S(t)$ at a given $t>0$.  In addition, prices of the underlying assets of the options in the portfolio are typically correlated, with a given correlation matrix $\Sigma$ (or its Cholesky factorization) for the corresponding Brownian motions. Without loss of generality, in this paper,  we demonstrate here quantum VaR/CVaR calculation assuming the constant elasticity volatility model~\cite{CEV75}, $\sigma(S,t)=\alpha/\sqrt{S}$, where $\alpha$ is a positive constant.

\subsection{PDE Approach to Compute Fair Option Value}
Assume that $r>0$ is the risk-free interest rate. For an option with expiry $T$, the fair option value function $V(S,t)$ is the unique solution to the partial differential equation:
 \begin{equation} \label{BSpde}
 \left\{
 \begin{split}
 &\frac{\partial V}{\partial t}  + \frac{1}{2} \sigma(S,t)^2 S^2 \frac{\partial^2 V}{\partial S^2} +
 r S\frac{\partial V}{\partial S} - rV = 0, \quad 0 \leq S < + \infty, ~ 0 \leq t <T\\
  & V(S,T) ={\rm payoff}(S),
 \end{split}
 \right.
 \end{equation}
which is an expression of the no-arbitrage assumption and the existence of a hedged portfolio under continuous rebalancing.
The expected rate of return of the underlying $\mu$ does not appear in BS PDE. Hence the no-arbitrage value of the option does not depend on $\mu$. It can be shown that, for an European option with an expiry $T>t>0$, under reasonable conditions, the fair value function has the form,
\begin{align}
V(S,t) = e^{-r(T-t)} \mathbb{E}^{Q}_{S,t}({\rm payoff}(S_{T}))
\end{align}
where $\mathbb{E}^Q_{S,t}(\cdot)$ is the conditional expectation, conditional on time $t$ and underlying price $S$, i.e., assuming $S_t$ now satisfies
\begin{equation}\label{eq:LVFRN}
\frac{dS(t)}{S(t)} = r {\rm \bm d}t + \sigma(S(t),t) {\rm \bm d}Z_t^Q ,
\end{equation}
where $Z_t^Q$ is a standard Brownian motion, and $r$ is the risk free interest rate.
Under a general model \eqref{eq:LVFReal}, this conditional expectation may not have an explicit formula but can be estimated using Monte-Carlo (MC) simulations.
%[\YW{shall we provide a mathematical expression on  $\mathbb{E}^Q_{S,t}(\cdot)$? the current notation is unclear}]. 
Note that the above BS PDE can be readily generalized to accommodate multi-assets.

The fair value of the American option is the solution to the following 
partial differential equation complementarity problem:
\begin{equation}\label{compl}
\left\{
\begin{array}{l}
\frac{\partial V}{\partial t} + \frac{1}{2}\sigma(S,t)^2 S^2
\frac{\partial^2 V}{\partial S^2}
+ r S \frac{\partial V}{\partial S} - r V \leq 0,\\
V(S,t)-\payoff(S) \geq 0 \\
(V(S,t)-\payoff(S))(
\frac{\partial V}{\partial t} + \frac{1}{2}\sigma(S,t)^2 S^2
\frac{\partial^2 V}{\partial S^2}
+ r S \frac{\partial V}{\partial S} - r V) = 0 \\
0\leq S < +\infty,~~~ 0< t <T
\\
V(S,T) = \payoff(S)
\end{array}
\right.
\end{equation}
 
There are two main approaches for computing fair option values, either PDE or MC, with the PDE offering more accurate values 
 %more easily 
 in general.
It is well known that the PDE approach can suffer the curse of dimensionality and becomes computationally infeasible on a classical computer as the dimensionality of the pricing problem becomes greater than $3$.
%[\YW{Ref?, YL: not necessary}].  
Monte Carlo simulation is an alternative approach for computing fair option value, where the underlying price is under the risk-neutral dynamics shown in Eq.~\eqref{eq:LVFRN}. However, using the Monte Carlo pricing method, computing VaR of a portfolio of options requires nested Monte Carlo simulations, which is also computationally expensive on classical computers. In addition, accurately computing the option value using MC faces more challenges for American options in comparison to European options.

\section{Outline of the Quantum Approach}
\label{Sec:outline}

While an option can have multiple assets as the underlying and typically simulations of joint distribution of all underlying prices are required, to highlight quantum implementation of option pricing, VaR calculation, and Monte Carlo simulation, here we consider the underlying of each option to be a single asset, which is shared by all options in the portfolio.

For each option $V(S,t)$ in the portfolio, option values on a finite grid can be computed by a finite difference PDE method. Recall that, for simplicity, here we have assumed that each option is written on a single asset. Assume a uniform discretization with timestep $\Delta\tau = T/\mathcal{N}$ along the time axis $t$ and a non-uniform discretization $\{S_0,S_1,\ldots, S_{2^n-1}\}$ along the  asset $S$ axis is provided. Let $V^t(S_j)=V(S_j,t)$ denote option values on the grid, where $j\in \{0,1, \ldots, 2^n-1\}$
%$j\in[2^n]$ 
and $t\in \mathcal{T}_{\rm set}=\{0,\Delta\tau, 2\Delta\tau,\cdots, \mathcal{N}\Delta\tau\}$.
Given a target time horizon $\bar{t}\in \mathcal{T}_{\rm set}$, 
%which is a multiple of $\Delta\tau$, 
%and error $\epsilon$, 
the proposed computation follows the following steps:
 
%\vspace*{-0.5cm}
% \begin{center}
% \parbox{6in}{
 \begin{itemize}
 \item 
Step~1 (to be implemented in Section~\ref{Sec:IV}): Option values $\vec{V}^{\bar{t}}=\{V^{\bar{t}}(S_j)\}_{j=0}^{2^n-1}$ are computed on the discretization grid described above by a PDE approach. This step generates the quantum state 
\begin{align}
    |V^{\bar t}\rangle=\sum_{j=0}^{2^n-1}\frac{V^{\bar t}(S_j)}{\sqrt{\sum_{j}(V^{\bar t}(S_j))^2}}|j\rangle,
    \label{Eq:Vt}
\end{align}
where $V^{\bar t}(S_j)$ represents the fair option value on the point $(\Bar t, S_j)$.

\item 
Step~2 (to be implemented in Section~\ref{Sec:V}): Stock price $S_{\rm sub}^{\bar{t}}=(S^{\Bar t}_1, S^{\Bar t}_2,\cdots, S^{\Bar t}_L)\in\mathbb{R}^L$ at time 
${\bar t}\in\mathcal{T}_{\rm set}$ are computed by Monte Carlo simulations in quantum parallel. Each Monte Carlo trajectory is shaped by a distinct sequence of pseudo-random numbers. This step prepares the quantum state $|\phi^{\bar t}\rangle=\frac{1}{\sqrt{L}}\sum_{k=1}^{L}|k\rangle|\tilde{S}^{\bar t}_k\rangle$, where the register $|\tilde{S}^{\Bar t}_k\rangle$ contains $m$ ancillary qubits such that $|\tilde{S}^{\Bar t}_k-S^{\Bar t}_k|\leq2^{-m}$ with $m=\log(1/\epsilon)$.

\item
Step~3 (to be implemented in Section~\ref{Sec:VI}): A finite sample distribution of option values is obtained by the following mapping
\begin{align*}
    U:\frac{1}{\sqrt{L}}\sum\limits_{k=1}^L|k\rangle_0|\tilde{S}^{\bar t}_k\rangle_1|0\rangle^m_2\mapsto |\Phi\rangle =\frac{1}{\sqrt{L}}\sum\limits_{k=1}^L|k\rangle_0|\tilde{S}^{\bar t}_k\rangle_1|V(\tilde{S}^{\bar t}_k)\rangle_2,
    \label{Eq:target}
\end{align*}
where $\tilde{S}^{\bar t}_k=S_{j^{\prime}}\in S$ for some index $j^{\prime}\in[2^n]$, and the corresponding option value $V(\tilde{S}^{\bar t}_k)=V^{\bar{t}}(S_{j^{\prime}})$.
\item 
Step~4 (to be implemented in Section~\ref{Sec:VII}): A bisection search on option values provides the VaR of the portfolio at time ${\bar t}$. Based on the estimated VaR, a quantum amplitude estimation algorithm can be used to obtain the CVaR.
\end{itemize}
%}
%\end{center}
We provide the technical details of the above four steps in the following sections.

%\FloatBarrier

\section{A Quantum Algorithm for European Option Values}
\label{Sec:IV}
\subsection{Discrete PDE for European Options}

For each option $V^{t}(S_j)$ in the portfolio, option values on a finite discretization grid can be computed by a finite different PDE method. Recall that, for simplicity, here we have assumed that each option is written on a single asset.
For simplicity, assume that volatility function does not depend on time, for example set $\sigma(S,t) = \alpha/\sqrt{S}$, and 
let 
\begin{eqnarray*}
 && \alpha_{j} = \left[ \frac{\sigma_j^2 S_j^2}
    {(S_j-S_{j-1})(S_{j+1}-S_{j-1}) }
       - \frac{r S_j}{S_{j} - S_{j-1}}  \right], \nonumber \\
 && \beta_{j} =
    \frac{\sigma_j^2 S_j^2} {(S_{j+1}-S_j)(S_{j+1} - S_{j-1})},
\end{eqnarray*}
where $r$ is the risk-free interest rate.
%$\alpha_j, \beta_j, r$ are formally defined in Appendix~\S\ref{data:pos}.
Let the vectors 
\begin{equation}
    V^{t} = \left[ \begin{array}{c}
        V^{t}(S_0) \\
        V^{t}(S_1) \\
         | \\
       V^{t}(S_{2^n-1})
       \end{array}
     \right],
\quad \quad
   V^{t-\Delta\tau} = \left[ \begin{array}{c}
        V^{t-\Delta\tau}(S_0) \\
        V^{t-\Delta\tau}(S_1) \\
         | \\
       V^{t-\Delta\tau}(S_{2^n-1})
       \end{array}
     \right],
\end{equation}
and let $M$ be the tridiagonal matrix with
entries
%\begin{align}
%     [M V^{t-\Delta\tau}]_j=-\Delta \tau \alpha_j V^{t-\Delta\tau}(S_{j-1})
%      +\Delta \tau (\alpha_j + \beta_j + r)
%               V^{t-\Delta\tau}(S_j)
%       - \Delta\tau \beta_j V^{t-\Delta\tau}(S_{j+1})
%       \label{Eq:Melement}
%\end{align}
\begin{align}
     [M V^{t-\Delta\tau}]_j=-\Delta \tau \alpha_j(V^{t-\Delta\tau}(S_{j-1})-V^{t-\Delta\tau}(S_{j}))
      -\Delta\tau\beta_j(V^{t-\Delta\tau}(S_{j+1})-V^{t-\Delta\tau}(S_{j}))
       +r\Delta\tau V^{t-\Delta\tau}(S_{j})
       \label{Eq:Melement}
\end{align}
when $j\neq 0, 2^n-1$. Let the first row and last row of $M$ correspond to boundary conditions.
We can write the fully implicit time stepping as
\begin{align}
  V^{t-\Delta\tau} =  [I_n + M]^{-1}V^{t}.
   \label{Eq:updaterule}
\end{align}
Updating $V^{t}$ to $V^{t-\Delta\tau}$ requires a high-dimensional matrix transformation especially for multi-asset option pricing, posing significant challenges for classical computers.
%which is believed to be a BQP-complete problem~\cite{harrow2009quantum}, it thus provides challenges for the classical computation model.

\subsection{A Quantum algorithm for European Option PDE}
Here, we demonstrate how to utilize a quantum computer to prepare the quantum state $|V^{\Bar{t}}\rangle$ given in Eq.~\ref{Eq:Vt} where  $\bar{t}$ represents the target time. 
%\YL{I deleted the word "boundary" since it is not really a boundary for PDE} 
According to the update rule (Eq.~\ref{Eq:updaterule}), the vector at $\Bar{t}$ can be represented by
\begin{align}
V^{\Bar{t}}=\left[I+M\right]^{-1}V^{\bar{t}+\Delta\tau}=\cdots=
\left[I+M\right]^{-(T-\bar{t})(\Delta\tau)^{-1}}V^{T},
\end{align}
where the starting point $V^{T}=[{\rm payoff}(S_0),\cdots, {\rm payoff}(S_{2^n-1})]$. Note that the matrix $M$ is a $3$-sparse tridiagonal matrix, whose entries are given in Eq.~\ref{Eq:Melement}, and is thus generally non-Hermitian. Then we need to consider the singular value decomposition of $\tilde{M}=I+M=\sum_{k=1}^J\lambda_k|w_k\rangle\langle v_k|$, where $J\leq\mathcal{O}(2^n)$ represents the rank of $\tilde{M}$, $(\lambda_1,\cdots,\lambda_J)$ represent singular values of $\tilde{M}$, and $\{|w_k\rangle\},\{|v_k\rangle\}_{k=1}^J$ are singular vectors of $\tilde{M}$. As a result, the vector $V^{\Bar{t}}$ can be further expressed in terms of the singular vectors $|w_k\rangle$ and $|v_k\rangle$, that is
\begin{align}
    |V^{\Bar{t}}\rangle=\tilde{M}^{-(T-\bar{t})\Delta\tau^{-1}}|V^T\rangle=\sum\limits_{k=1}^J\frac{\langle v_k|V^T\rangle}{\lambda_k^{(T-\bar{t})\Delta\tau^{-1}}}|w_k\rangle.
\end{align}
Since the function ${\rm payoff}(\cdot)$ is classically efficiently computable, the initial quantum state $$|V^T\rangle=\sum\limits_{i=0}^{2^n-1}\frac{{\rm payoff}(S_i)|i\rangle}{\sqrt{\sum_{i=0}^{2^n-1}{\rm payoff}(S_i)^2}}$$ can be efficiently prepared by the Grover-Randolph algorithm~\cite{grover2002creating}. 
In the following, we demonstrate how to utilize the Quantum Singular Value Transformation~(QSVT)~\cite{gilyen2019quantum} to simulate the process $\tilde{M}^{-(T-\bar{t})\Delta\tau^{-1}}|V^T\rangle$ on a quantum computer and finally yield the quantum state $|V^{\bar{t}}\rangle$.

Let $g(x)=\frac{1}{2}(x/\|\tilde{M}\|_2)^{-(T-\bar{t})\Delta\tau^{-1}}$, where the norm $\|\tilde{M}\|_2=\max\{\abs{\lambda_k}\}_{k=1}^J$, then the QSVT method provides a quantum circuit implementation to achieve $\frac{g(\tilde{M})|V^T\rangle}{\|g(\tilde{M})|V^T\rangle\|}$, where the matrix function $g(\tilde{M})=\sum_{r=1}^Jg(\lambda_r)|w_r\rangle\langle v_r|$. Here, the fundamental idea within the QSVT method relies on two steps:
\begin{itemize}
\item Encoding the non-unitary matrix $\tilde{M}$ into a higher-dimensional unitary matrix $U_{\tilde{M}}$~(Encoding Phase);
\item Constructing a $d$-degree polynomial approximations $P(x)$ to $g(x)$ within $\epsilon$-additive error. Furthermore, finding phase factors $\Phi=(\phi_1,\cdots,\phi_d)\in\mathbb{R}^d$ to modulate $U_{\tilde{M}}$ to achieve $P(\tilde{M})$ which approximates $g(\tilde{M})$~(Modulation Phase).
\end{itemize}

\subsubsection{Encoding Phase}
We introduce the block-encoding technique, where the fundamental idea is to represent a sub-normalized matrix as the upper-left block of a unitary matrix
\begin{equation}
 U_{\tilde{M}}=
\begin{bmatrix} 
\tilde{M}/\gamma & \cdots\\
\cdots & \cdots
\end{bmatrix}
\end{equation}
which is equivalent to $\tilde{M}=\gamma\left(\langle0|\otimes I\right)U_{\tilde{M}}\left(|0\rangle\otimes I\right)$.

\begin{definition}[Block-Encoding]
    Suppose that the matrix $\tilde{M}$ is an $n$-qubit operator, $\gamma, \epsilon\in\mathbb{R}_{+}$ and $a\in\mathbb{N}$, then we say that the $(n+a)$-qubit unitary $U_{\tilde{M}}$ is an $(\gamma,a,\epsilon)$-block-encoding of $\tilde{M}$, if
    \begin{align}
        \|\tilde{M}-\gamma\left(\langle0|^{\otimes a}\otimes I_n\right)U_{\tilde{M}}\left(|0\rangle^{\otimes a}\otimes I_n\right)\|\leq \epsilon.
    \end{align}
\end{definition}
In our case, $M$ is a tridiagonal matrix, and $\tilde{M}=M+I_n\in\mathbb{R}^{2^n\times 2^n}$ is thus a $3$-sparse matrix. Camps et al.~\cite{camps2022explicit} provided an efficient method for constructing a Block-Encoding of sparsity matrices. Consider $c(j,l)$ to be a function that gives the row index of the $l$-th non-zero matrix elements in the $j$-th column of $\tilde{M}\in\mathbb{R}^{2^n\times 2^n}$. Suppose there exists a unitary $U_c$ such that $U_c|l\rangle|j\rangle=|l\rangle|c(j,l)\rangle$, and a unitary $U_R$ such that
\begin{align}
U_R|0\rangle|l\rangle|j\rangle=\left(\tilde{M}_{c(j,l),j}|0\rangle+\sqrt{1-|\tilde{M}_{c(j,l),j}|^2}|1\rangle\right),
\end{align}
then $U_{\tilde{M}}=\left(I_2\otimes H^{\otimes 2}\otimes I_2^{\otimes n}\right)(I_2\otimes U_c)U_R\left(I_2\otimes H^{\otimes 2}\otimes  I_2^{\otimes n}\right)$ represents a $(1,3,0)$-block-encoding of $\tilde{M}/3$. This statement can be verified easily. Starting from the quantum state $|0\rangle|0^{[\log s]}\rangle|j\rangle$, we have
\begin{eqnarray}
\begin{split}
    |0\rangle|0^{[\log s]}\rangle|j\rangle
     &\stackrel{\left(I_2\otimes H^{\otimes 2}\otimes  I_2^{\otimes n}\right)}{\longrightarrow} \frac{1}{\sqrt{s}}\sum\limits_{l\in[s]}|0\rangle|l\rangle|j\rangle
     \stackrel{U_R}{\longrightarrow}\frac{1}{\sqrt{s}}\sum\limits_{l\in[s]}\left(\tilde{M}_{c(j,l),j}|0\rangle+\sqrt{1-|\tilde{M}_{c(j,l),j}|^2}|1\rangle\right)|l\rangle|j\rangle\\
     &\stackrel{\left(I_2\otimes U_c\right)}{\longrightarrow} \frac{1}{\sqrt{s}}\sum\limits_{l\in[s]}\left(\tilde{M}_{c(j,l),j}|0\rangle+\sqrt{1-|\tilde{M}_{c(j,l),j}|^2}|1\rangle\right)|l\rangle|c(j,l)\rangle\\
     &\stackrel{\left(I_2\otimes H^{\otimes 2}\otimes  I_2^{\otimes n}\right)}{\longrightarrow} \frac{1}{\sqrt{s}}\sum\limits_{l\in[s]}\left(\tilde{M}_{c(j,l),j}|0\rangle+\sqrt{1-|\tilde{M}_{c(j,l),j}|^2}|1\rangle\right)H^{\otimes 2}|l\rangle|c(j,l)\rangle.
\end{split}
\end{eqnarray}
Finally, we utilize $\langle0|\langle0^{[\log s]}|\langle i|$ to post-process the above quantum state, hence the inner product $$\langle0|\langle0^{\log s}|\langle i|U_{\tilde{M}}|0\rangle|0^{\log s}\rangle|j\rangle=\frac{\tilde{M}_{ij}}{s}$$ which implies $U_{\tilde{M}}$ is a $(1,3,0)$-block-encoding of $\tilde{M}/s$.
%\YL{Yusen/Jingbo: it is not clear to me how to interpret q=2 interpreted here}

%\YW{Yuying, the number of ancilla qubits is depended by the sparsity of the target matrix. Since the sparsity of $\tilde{M}$ is $3$, as a result $2$ ancilla qubits suffice.}

\subsubsection{Modulation Phase}
A polynomial approximation of $g(x)=\frac{1}{2}(x/\|\tilde{M}\|_2)^{-(T-\bar{t})\Delta\tau^{-1}}$ is required.
\begin{lemma}
    Let $\epsilon\in(0,1/2]$, $T-\bar{t}>0$ and function $g(x)=\frac{1}{2}(x/\|\tilde{M}\|_2)^{-(T-\bar{t})\Delta\tau^{-1}}$, then there exist a real polynomial function $P(x)\in\mathbb{R}[x]$, such that $\|P(x)-g(x)\|_{[\|\tilde{M}\|_2^{-1},1]}\leq\epsilon$ and $\|P(x)\|_{[-1,1]}\leq 1$. The degree of $P$ is at most $\mathcal{O}\left((T-\bar{t})\Delta\tau^{-1}\|\tilde{M}\|_2\log(1/\epsilon)\right)$.
\label{lemma:polynomial_approximation}
\end{lemma}
The proof is based on Corollaries~66 in Ref~\cite{gilyen2019quantum}. 
\begin{lemma}[Ref~\cite{gilyen2019quantum}]
    Let $x_0\in[-1,1]$, $r\in(0,2], \delta\in(0,\eta]$ and let $f:[-x_0-\eta-\delta,x_0+\eta+\delta]\mapsto\mathbb{C}$ and  $f(x_0+x)=\sum_{l=0}^{\infty}a_lx^l$ for all $x\in[-\eta-\delta,\eta+\delta]$. Suppose that $\sum_{l=0}^{\infty}(\eta+\delta)^l|a_l|\leq B$. Let $\epsilon\in[0,1/2B]$, then there is an efficiently computable polynomial $P(x)\in\mathbb{C}[x]$ of degree $J\leq\mathcal{O}(\frac{1}{\delta}\log(B/\epsilon))$ such that
    $\|f(x)-P(x)\|_{[x_0-\eta,x_0+\eta]}\leq\epsilon$.
    \label{lemma:lemma2}
\end{lemma}
%\YL{Yusen: $r$ has been used to denote interest rate in  Section II, e.g., (2) etc. Can you use a different symbol here for $r$?}
\begin{proof}[Proof of Lemma~\ref{lemma:polynomial_approximation}]
    Denote $\tilde{T}=(T-\tilde{t})\Delta\tau^{-1}$. For all $y\in(-1,1)$, the generalized binomial theorem yields $(1+y)^{-\tilde{T}}=\sum_{k=0}^{\infty}\tbinom{-\tilde{T}}{k}y^k$, where $\tbinom{-\tilde{T}}{k}=-\tilde{T}(-\tilde{T}-1)\cdots(-\tilde{T}-k+1)/k!$. Let $x_0=0, \eta=1-\|\tilde{M}\|_2^{-1}$ and $\delta=\tilde{T}^{-1}\|\tilde{M}\|_2^{-1}$. Furthermore, we have $a_l=\frac{\|\tilde{M}\|_2^{\tilde{T}}}{2}\tbinom{-\tilde{T}}{k}$ and
    \begin{align}
        \sum\limits_{l=0}^{\infty}(\eta+\delta)^l|a_l|=\frac{\|\tilde{M}\|_2^{\tilde{T}}}{2}\sum\limits_{l=0}^{\infty}(\eta+\delta)^l\abs{\tbinom{-\tilde{T}}{l}}=\frac{\|\tilde{M}\|_2^{\tilde{T}}}{2}\sum\limits_{l=0}^{\infty}\tbinom{-\tilde{T}}{l}(-\eta-\delta)^l=\frac{\|\tilde{M}\|_2^{\tilde{T}}}{2}(1-\eta-\delta)^{-\tilde{T}}\leq\frac{e}{2}.
    \end{align}
 As a result, the upper bound of $\sum_{l=0}^{\infty}(\eta+\delta)^l|a_l|$ is $B=e/2$, and the polynomial approximation $P(x)$ 
 %\begin{align}
  %   P(x)=\sum\limits_{l=0}^{J-1}a_l(r+\delta)^l(x+1)^l=\sum\limits_{l=0}^{J-1}\frac{\|\tilde{M}\|_2^T}{2}\tbinom{-T}{l}\left(1-\|\tilde{M}\|_2^{-1}+T^{-1}\|\tilde{M}\|_2^{-1}\right)^l(x+1)^l
 %\end{align}
 has degree $\leq\mathcal{O}(\delta^{-1}\log(B/\epsilon))=\mathcal{O}\left(\tilde{T}\|\tilde{M}\|_2\log(1/\epsilon)\right)\mathcal{O}\left((T-\bar{t})\Delta\tau^{-1}\|\tilde{M}\|_2\log(1/\epsilon)\right)$.
\end{proof}

%In the context of QSVT framework, singular values of $\tilde{M}$ should be normalized into the interval $[-1,1]$. Let $\delta=\|\tilde{M}\|^{-1}_{*}$, then Lemma~\ref{lemma:polynomial_approximation} provides a $\mathcal{O}\left(T\|\tilde{M}\|_2\log(1/\epsilon)\right)$-degree polynomial approximation $P(\tilde{M})$ to $g(\tilde{M})=\frac{1}{2}\tilde{M}^{-T}$ within $\epsilon$ additive error. The quantum signal processing method 

%\subsection{Modulation Phase}
\begin{fact}[QSVT]
Suppose the matrix $\tilde{M}\in \mathbb{R}^{2^n\times 2^n}$ is encoded by a $(1,2,0)$-block-encoding unitary $U_{\tilde{M}}$. Given the polynomial function $P(x)$ described in Lemma~\ref{lemma:polynomial_approximation}, there exists a set of phase factors $\Phi=(\phi_1,\cdots,\phi_d)$ such that 
    \begin{equation}
        U_{\Phi}=\left\{
        \begin{aligned}
            e^{i\phi_1(2\Pi-I)}U_{\tilde{M}}\prod\limits_{j=1}^{(d-1)/2}\left(e^{i\phi_{2j}(2\Pi-I)}U_{\tilde{M}}^{\dagger}e^{i\phi_{2j+1}(2\Pi-I)}U_{\tilde{M}}\right), (\text{$d$ is odd})\\
            \prod\limits_{j=1}^{d/2}\left(e^{i\phi_{2j}(2\Pi-I)}U_{\tilde{M}}^{\dagger}e^{i\phi_{2j+1}(2\Pi-I)}U_{\tilde{M}}\right), (\text{$d$ is even})
        \end{aligned}
        \right.
        \label{Eq:QSVT}
    \end{equation}
and $U_{\Phi}$ is a $(1,4,0)$-block-encoding unitary of $P(\tilde{M})$, where $d=\mathcal{O}\left((T-\bar{t})\Delta\tau^{-1}\|\tilde{M}\|_2\log(1/\epsilon)\right)$ and the projector $\Pi=\left(|0\rangle\langle0|\right)^{\otimes 3}$.
\end{fact}
Noting that the phase factors $(\phi_1,\cdots,\phi_d)$ can be efficiently calculated by using the QSPPACK. The above result directly yields 
\begin{align}
    \|\left(\langle10^{\otimes 3}|\otimes I_n\right)U_{\Phi} \left(|10^{\otimes 3}\rangle\otimes I_n\right)|V^T\rangle-|V^{\bar{t}}\rangle\|_2\leq\epsilon,
\end{align}
where the query complexity of $U_{\Phi}$ is at most $\mathcal{O}\left((T-\bar{t})\Delta\tau^{-1}\|\tilde{M}\|_2\log(1/\epsilon)\right)$.

\begin{figure*}[h]
\centering
\includegraphics[width=0.95\textwidth]{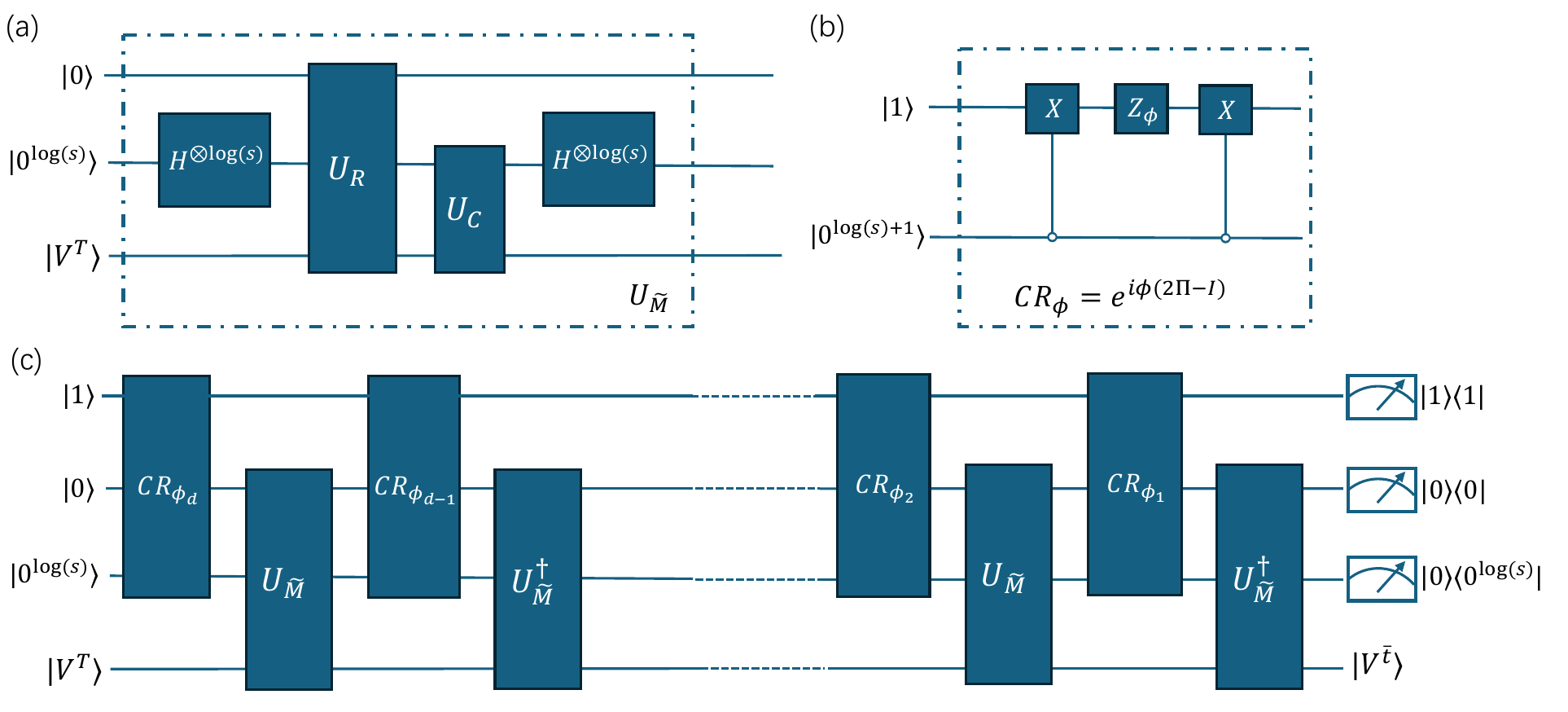} 
\caption{\textbf{Quantum circuit in implementing Step~1.} (a) A quantum circuit to achieve the $(1,3,0)$-block-encoding of the matrix $\tilde{M}$. (b) The quantum circuit to achieve the controlled rotation ${\rm CR}_{\phi}$. (c) A quantum circuit of depth $d=\mathcal{O}\left((T-\bar{t})\Delta\tau^{-1}\|\tilde{M}\|_2\log(1/\epsilon)\right)$ to achieve the $(1,4,0)$-block-encoding of the polynomial transformation $g(\tilde{M})$, as shown in Eq.~\ref{Eq:QSVT}. Suppose $g(\cdot)$ can be approximated by a $d$-degree polynomial function within $\epsilon$ additive error, then $U_{\Phi}$ is composed of $d$ controlled rotations ${\rm CR}_{\phi}$ and a $d$ block-encoding of $\tilde{M}$.}
\label{fig:1}
\end{figure*}

\section{A Quantum Circuit for Monte Carlo Stock Price Simulation in Quantum Parallel}
\label{Sec:V}
Underlying price dynamics can be modelled according to a generalized Black-Scholes framework given by Eq.(\ref{eq:LVFReal}), which includes stochastic volatility using a local volatility function. This enhanced model demonstrates improved capability in matching traded option market prices, i.e. generating a volatility smile, which is widely recognized and documented in the equity options markets \cite{Dup94,Der94,Dum98,CLV98}. The local volatility function allows volatility to vary with time and future random stock price realization, potentially leading to a more accurate depiction of price uncertainty. 
This enhanced model provides a sophisticated tool for analyzing and predicting stock price behaviour, recognizing the dynamic and variable nature of volatility in financial markets. Applying Euler's method to \eqref{eq:LVFReal},  we have
\begin{align}
\label{stochevolve}
    S^{t+\Delta\tau}=S^{t}\left(1 + \mu \Delta\tau + \frac{\alpha}{\sqrt{S^{t}}} \Delta Z_t\right),
\end{align}
where $\Delta Z_t$ is the change in a standard Brownian motion, $\mu$ is a drift, $\sigma(S,t)=\frac{\alpha}{\sqrt{S}}$ is a local volatility function and the time variable $t\in[0,\bar{t}]$ which is different to the backward propagation shown in Sec.~\ref{Sec:IV}. Let $S^0_{\rm sub}=(S^0_1,\cdots,S^0_L)$,
$S^{t+\Delta\tau}_j=S^{t}_j\left(1+\mu\Delta\tau+\alpha(\Delta Z)^t_j(S^{t}_j)^{-1/2}\right)$, $\tilde{S}_j^t$ represent the estimation of $S_j^t$ within $2^{-m}$ additive error and the number of ancillary qubits $m=\mathcal{O}(\log(1/\epsilon)$. In the following, we demonstrate how to utilize a quantum computer to prepare $|\phi^{\bar t}\rangle=L^{-1/2}\sum_{k=1}^{L}|k\rangle|\tilde{S}^{\bar t}_k\rangle$ in parallel.

The quantum algorithm starts from the initial state
\begin{align}
    |S^0_{\rm sub}\rangle=\left(\frac{1}{\sqrt{L}}\sum\limits_{j=1}^L|j\rangle_1\right)\otimes |\tilde{S}^0\rangle_2
\end{align}
where the first register contains $\log(L)$ qubits while the second register has $m$ qubits. Then $|S^0_{\rm sub}\rangle$
can be efficiently prepared by $[\log(L)]$-Hadamard gates and $\mathcal{O}(m)$ Pauli-X gates, since the composite system $1$ and $2$ are in a tensor product state. To simulate the random variable $\Delta Z_j$ in the quantum circuit, we assign $\Delta Z_j=4j/L(1-j/L)$ which naturally simulates the logistic chaos variable. Let the function $\mathcal{F}(j,x)=(1+\mu\Delta\tau)x+\alpha\left(\Delta Z_j\right)\sqrt{x}$ and define the quantum gate 
\begin{align}
    {\rm CC}\text{-}{\rm CU}_{S}: |j\rangle|x\rangle|0\rangle^m\mapsto |j\rangle|x\rangle|\mathcal{F}(j,x)\rangle.
\end{align}
%\YL{Yusen: should the superscript in $\tilde{S}_j^{t-1}$ here, including Figure 2, be updated?}
This enables us to achieve 
\begin{align}
   \frac{1}{\sqrt{L}}\sum\limits_{j=1}^L|j\rangle|\tilde{S}_j^{t-\Delta \tau}\rangle|0\rangle^m\stackrel{{\rm CC}\text{-}{\rm CU}_{S}}{\longrightarrow} \frac{1}{\sqrt{L}}\sum\limits_{j=1}^L|j\rangle|\tilde{S}_j^{t-\Delta \tau}\rangle|\tilde{S}_j^t\rangle.
\end{align}
Note that the inverse function of $\mathcal{F}(\cdot)$ exists, that is 
$\mathcal{F}^{-1}(j,x)=\left(\sqrt{(a^{-1}(x+b^2/4a))}-b/2a\right)^2$, where $a=(1+\mu\Delta\tau)$ and $b=\alpha\Delta Z_j$. This thus enables  the map 
\begin{align}
    {\rm CC}\text{-}{\rm CU}^{-1}_{S}: |j\rangle|\tilde{S}_j^{t-\Delta \tau}\rangle|\tilde{S}_j^{t}\rangle\mapsto |j\rangle|\tilde{S}_j^{t-\Delta \tau}\oplus \mathcal{F}^{-1}(j,\tilde{S}_j^{t}) \rangle|\tilde{S}_j^{t}\rangle.
\end{align}
Using quantum control gates ${\rm CC}$-${\rm CU}_{S^t}$ and ${\rm CC}$-${\rm CU}^{-1}_{S^t}$ iteratively, the initial state  $|\vec{S}^0\rangle$ may evolves to 
\begin{align}
    |\phi^{\bar{t}}\rangle=\frac{1}{\sqrt{L}}\sum\limits_{j=1}^L|j\rangle|\tilde{S}_j^{\bar{t}}\rangle
\end{align}
after $\bar{t}\Delta\tau^{-1}$ steps. The elaborate quantum circuit is shown in Fig.~\ref{fig:2}.

\begin{figure*}[h]
\centering
\includegraphics[width=0.80\textwidth]{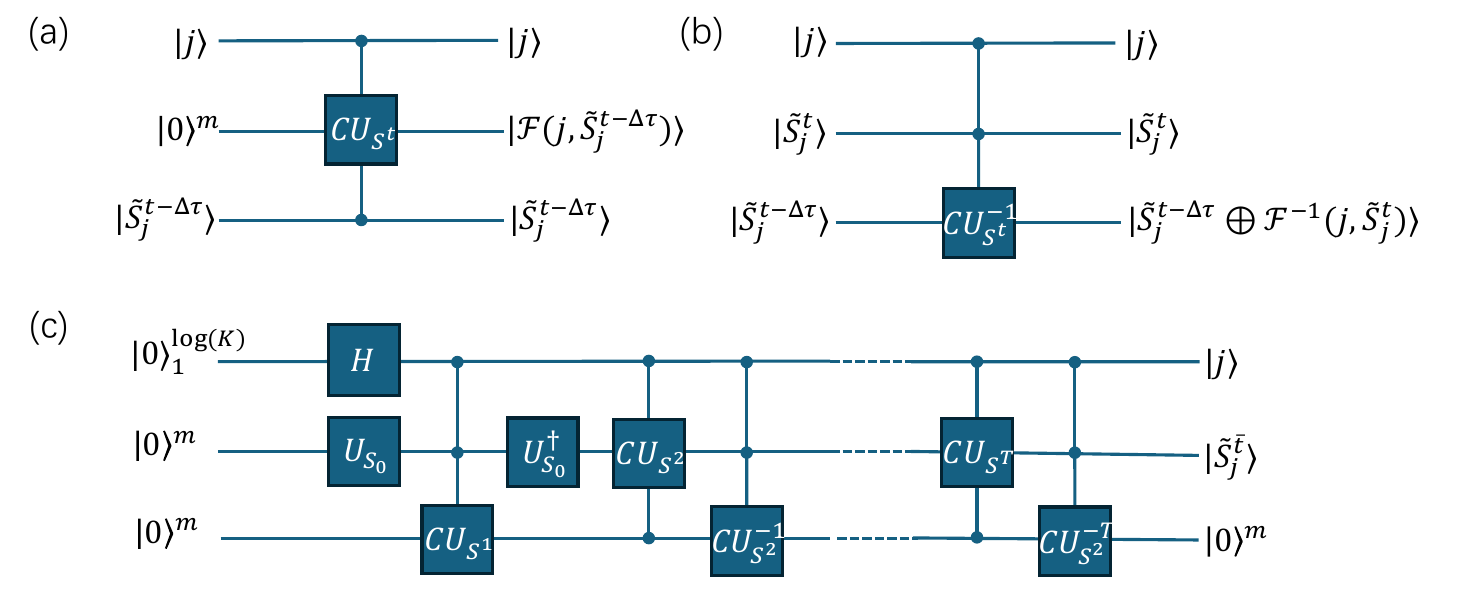} 
\caption{
\textbf{Quantum Circuit for Sub-Algorithm~2} 
(a) ${\rm CC}$-${\rm CU}_{S^t}$ maps $|j\rangle|0\rangle^m|\tilde{S}_j^{t-\Delta \tau}\rangle\mapsto |j\rangle|\mathcal{F}(j,\tilde{S}_j^{t-\Delta \tau})\rangle|\tilde{S}_j^{t-\Delta \tau}\rangle$, where the function $\mathcal{F}(j,\tilde{S}_j^{t-\Delta \tau})=(1+\mu\Delta\tau)\tilde{S}_j^{t-\Delta \tau}+\alpha\Delta Z_j\sqrt{\tilde{S}_j^{t-\Delta \tau}}$. (b) When $1+\mu\Delta\tau>0$, the function 
$\mathcal{F}^{-1}(j,\tilde{S}_j^{t})$
exists and equals $\mathcal{F}^{-1}(j,\tilde{S}_j^{t})=\tilde{S}_j^{t-\Delta \tau} = \left[\sqrt{(a^{-1}(\tilde{S}_j^{t}+b^2/4a))}-b/2a\right]^2$, 
where $a=(1+\mu\Delta\tau)$ and $b=\alpha\Delta Z_j$. The controlled-unitary ${\rm CC}$-${\rm CU}^{-1}_{S^t}$ %essentially  
then maps $|j\rangle|\tilde{S}_j^{t}\rangle|\tilde{S}_j^{t-\Delta \tau}\rangle\mapsto |j\rangle|\tilde{S}_j^{t}\rangle|\tilde{S}_j^{t-\Delta \tau}\oplus \mathcal{F}^{-1}(j,\tilde{S}_j^{t}) \rangle$. (c) This circuit maps $|S_{\rm sub}^0\rangle\mapsto |S_{\rm sub}^{\bar{t}}\rangle$, where $U_{S_0}|0^m\rangle=|S^0\rangle$. }
\label{fig:2}
\end{figure*}
%Suppose we have the initial state $|\phi_0\rangle=|0\rangle_1^n|\tilde{S}_0\rangle_2$, where the second register has $m$ qubits such that $m>n$.

\section{A Quantum Algorithm for Distribution of Value of Portfolio}
\label{Sec:VI}
%\YL{I changed the title here. Please check if you are OK}
Given several copies of the input quantum state $|V^{\Bar{t}}\rangle=\sum_{i\in\{0,1\}^n}\tilde{V}^{\Bar{t}}(S_i)|i\rangle$, we demonstrate how to achieve the map
\begin{align}
   U:|\phi^{\bar{t}}\rangle=\frac{1}{\sqrt{L}}\sum\limits_{k=1}^L|k\rangle_0|\bm\tilde{S}_k^{\bar{t}}\rangle_1|0\rangle^m_2\mapsto \frac{1}{\sqrt{L}}\sum\limits_{k=1}^L|k\rangle_0|\bm\tilde{S}_k^{\bar{t}}\rangle_1|\tilde{V}^{\Bar{t}}(\bm\tilde{S}_k)\rangle_2,
    \label{Eq:target}
\end{align}
where $\{\bm\tilde{S}^{\bar{t}}_1,\cdots,\bm\tilde{S}^{\bar{t}}_L\}$ represents the stock price at time $\bar{t}$.

The fundamental idea is based on Quantum Principle Component Analysis~(QPCA)~\cite{lloyd2014quantum,yu2019quantum} and Quantum Phase Estimation~(QPE)~\cite{nielsen2002quantum}. Suppose we have prepared the quantum state $|V^{\Bar{t}}\rangle$, then added $m$ ancillary qubits to store the grid information $S_j$ in the ancillary register
\begin{align}
    |V^{\Bar{t}}\rangle|0\rangle^m\mapsto |\psi_2\rangle=\frac{1}{\sqrt{\sum_{j}(V^{\Bar{t}}(\tilde{S}_j))^2}}\sum_{j\in\{0,1\}^n}V^{\Bar{t}}(S_j)|j\rangle_1|\tilde{S}_j\rangle_2.
    \label{Eq:LoadGridInformation}
\end{align}
Discarding subsystem $1$, the quantum system naturally becomes the density matrix
\begin{align}
     \rho={\rm Tr}_{1}\left[|\psi_2\rangle\langle\psi_2|\right]=\frac{\sum_{i\in\{0,1\}^n}(V^{\Bar{t}}(S_i))^2|\tilde{S}_i\rangle_1\langle \tilde{S}_i|_1}{\sum_{i\in\{0,1\}^n}(V^{\Bar{t}}(S_i))^2}.
     \label{Eq:density_matrix}
\end{align}
We can now utilize the QPCA method~\cite{lloyd2014quantum} to simulate $e^{-i\rho \tau}$ and extract the spectrum information of $\rho$ by using the QPE algorithm~\cite{nielsen2002quantum}. To achieve the estimation $|\tilde{S}_j\rangle_1|0\rangle^m\mapsto|\tilde{S}_j\rangle_1|\hat{V}^{\Bar{t}}(\tilde{ S}_j)\rangle$, where the estimated value $\hat{V}^{\Bar{t}}(\tilde{S}_j)$ satisfies $$\abs{\hat{V}^{\Bar{t}}(\tilde{S}_j)-\frac{V^{\Bar{t}}(\tilde{S}_j)}{\sqrt{\sum_{j}(V^{\Bar{t}}(S_j))^2}}}\leq 2^{-m},$$ we need the time parameter $\tau=\mathcal{O}(2^{m})$.

Starting from the quantum state $\left(\mathbb{I}_0\otimes\mathbb{I}_1\otimes H^{\otimes m}\right)|\phi^{\bar{t}}\rangle=\frac{1}{\sqrt{L}}\sum_{k=1}^L|k\rangle_0|\bm\tilde{S}_k^{\bar{t}}\rangle_1\otimes H^m|0\rangle^m_2$, the QPCA algorithm can help us estimate the spectrum information of $\rho$ from the ancillary qubits. For any $n$-qubit density matrix $\sigma$, we have ${\rm Tr}_{A}[e^{-i\omega\Delta t}\rho_A\otimes\sigma_B e^{i\omega\Delta t}]=e^{-i\rho\Delta t}\sigma e^{i\rho\Delta t}+\mathcal{O}((\Delta t)^2)$, where $\omega$ represents the $2n$-qubit swap operator. As a result, we can perform the quantum operation 
\begin{align}
    \mathbb{I}_0\otimes \sum\limits_{l=1}^{N_{\rm qpe}}e^{-i\rho l\Delta t}\otimes |l\Delta t\rangle_2\langle l\Delta t|_2
    \label{Eq:qpe}
\end{align}
on systems $1, 2$, then perform the inverse Quantum Fourier Transformation and square-root function to prepare the quantum state 
\begin{align}
     |\Phi\rangle=\frac{1}{\sqrt{L}}\sum\limits_{k=1}^L|k\rangle_0|\bm\tilde{S}_k^{\bar{t}}\rangle_1|\hat{V}^{\Bar{t}}(\bm\tilde{S}_k)\rangle_2.
     \label{Eq:target_state}
\end{align}
To provide an estimation of $\hat{V}^{\bar{t}}(\bm\tilde{S}_k)$ within an acceptable accuracy, we need the evolution time $\tau=N_{\rm qpe}\Delta t=\mathcal{O}(2^{m})$. The whole process is visualized as Fig.~\ref{fig:3}.

\begin{figure*}[h]
\centering
\includegraphics[width=0.95\textwidth]{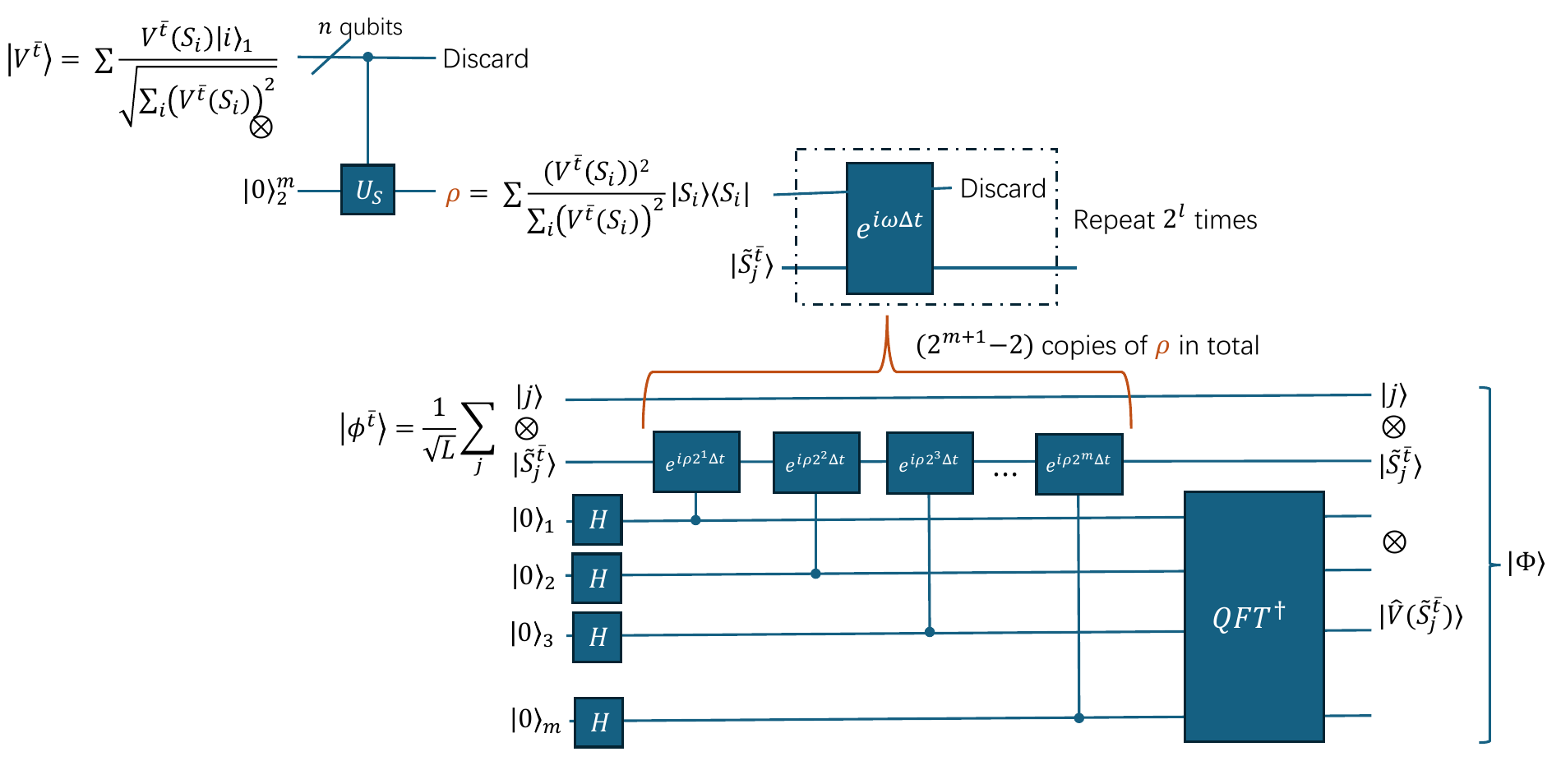} 
\caption{\textbf{Quantum Circuit for computing the distribution of portfolio value.} 
A quantum circuit in preparing the density matrix $\rho$ which is shown in Eq.~\ref{Eq:density_matrix} is shown in the top-left, where the quantum gate $U_S$ achieves the map $|i\rangle|0\rangle^m\mapsto |i\rangle|S_i\rangle$. To extract information hidden in the density matrix $\rho$, the QPCA is used in simulating the unitary operator $e^{i\rho2^l\Delta t}$, achieving the quantum phase estimation required for the mapping described by Eq.~\ref{Eq:target}.
}
\label{fig:3}
\end{figure*}

\section{A Quantum Algorithm for VaR/CVaR estimation}
\label{Sec:VII}
\subsection{VaR Estimation}
Here, we try to find the smallest value $\bar{V}\in\left[\min_{j\in[L]}(\hat{V}^{\bar{t}}(\tilde{S}_j)),\max_{j\in[L]}(\hat{V}^{\bar{t}}(\tilde{S}_j))\right]$ such that ${\rm Pr}_{j\in[L]}(\hat{V}^{\bar{t}}(\tilde{S}_j)\leq\Bar{V})=0.05$. To find $\bar{V}$ on a quantum computer, we utilize the bisection search method over $\bar{V}$~\cite{woerner2019quantum}. The bisection search approach depends on the unitary $U_{\rm CC}$ that achieves the map $$|\hat{V}^{\bar{t}}(\tilde{S}_j)\rangle|\Bar{V}\rangle|0\rangle\mapsto|\hat{V}^{\bar{t}}(\tilde{S}_j)\rangle|\Bar{V}\rangle|h\rangle,$$ where 
$h=0$ if
$\hat{V}^{\bar{t}}(\tilde{S}_j)\leq \Bar{V}$, otherwise $h=1$. Applying $U_{\rm CC}$ to the quantum state $|\Phi\rangle|\bar{V}\rangle|0\rangle$ produces
 \begin{align}
        \sqrt{p(0)}\sum\limits_{\hat{V}^{\bar{t}}(\tilde{S}_k)\leq\bar{V}}|k\rangle_0|\tilde{S}_k\rangle_1|\hat{V}^{\bar{t}}(\tilde{S}_j)\rangle_2|\bar{V}\rangle_3|0\rangle_4+\sqrt{p(1)}\sum\limits_{\hat{V}^{\bar{t}}(\tilde{S}_j)>\bar{V}}|k\rangle_0|\tilde{S}_k\rangle_1|V^{\bar{t}}(\tilde{S}_k)\rangle_2|\bar{V}\rangle_3|1\rangle_4.
    \end{align}
The probability of measuring $|0\rangle$ for the last qubit is $p(0)={\rm Pr}_{j\in[L]}(\hat{V}^{\bar{t}}(\tilde{S}_j)\leq\Bar{V})$. Therefore, with a bisection search over $\bar{V}$, we may efficiently find the smallest $\bar{V}$ such that ${\rm Pr}_{j\in[L]}(\hat{V}^{\bar{t}}(\tilde{S}_j)\leq\Bar{V})=0.05$ in at most $m$ steps, where $m$ represents the number of qubits in representing $\hat{V}^{\bar{t}}(\tilde{S}_j)$ given in the second register. Using the quantum mean value estimation~\cite{montanaro2015quantum}, we may estimate a $\epsilon$-approximation to $p(0)$ within $\mathcal{O}(1/\epsilon)$ queries to $|\Phi\rangle$ (as defined by Eq.~\ref{Eq:target_state}).

\subsection{CVaR Estimation}

%\YL{
%\textbf{CVaR:}
%Assume that the loss distribution is continuous.
% Conditional value at risk, CVaR,  is the average $P\&L$, given $P\&L$ is  less than VaR, i.e.,
%\begin{align}
%    \mathbb{E} [( P\&L ~ | P\&L \leq \mbox{VaR}_{\beta})]
%\end{align}
%Here $P\&L$ equals $V^{\bar{t}}$. For discrete samples, CVaR can be computed as the average below VaR.

Suppose we have estimated the VaR value $\bar{V}$ such that ${\rm Pr}_{j\in[L]}(\hat{V}^{\bar{t}}(\tilde{S}_j)\leq\bar{V})=5\%$. Then we can divide all $\hat{V}^{\bar{t}}(\tilde{S}_j)$ given in Eq.~\ref{Eq:target_state} into two sets: $\mathcal{L}_1=\{\hat{V}^{\bar{t}}(\tilde{S}_j)|\hat{V}^{\bar{t}}(\tilde{S}_j)\leq \bar{V}\}$ and $\mathcal{L}_2=\{\hat{V}^{\bar{t}}(\tilde{S}_j)|\hat{V}^{\bar{t}}(\tilde{S}_j)>\bar{V}\}$. As a result, CVaR could be estimated by
\begin{align}
    \frac{\sum_{\hat{V}^{\bar{t}}(\tilde{S}_j)\in\mathcal{L}_1}\hat{V}^{\bar{t}}(\tilde{S}_j)}{\|\mathcal{L}_1\|},
\end{align}
where $\|\mathcal{L}_1\|$ represents the number of entries in $\mathcal{L}_1$. 

%\YL{Yusen and Jingbo: A different expression is used in \eqref{eq:CVaR}?}

The quantum algorithm in predicting CVaR can be summarized as follows.
\begin{itemize}
    \item After estimating the VaR of $\bar{V}$, we append ancillary qubits $|\bar{V}\rangle_3|0\rangle_4$ to $|\Phi\rangle$ (as defined by Eq.~\ref{Eq:target_state}), and perform the unitary $U_{\rm CC}$ to achieve
    \begin{align}
|\Phi^1\rangle=\sqrt{\frac{\|\mathcal{L}_1\|}{L}}\sum\limits_{\hat{V}^{\bar{t}}(\tilde{S}_k)\in\mathcal{L}_1}|k\rangle_0|\tilde{S}_k\rangle_1|\hat{V}^{\bar{t}}(\tilde{S}_j)\rangle_2|\bar{V}\rangle_3|0\rangle_4+\sqrt{\frac{\|\mathcal{L}_2\|}{L}}\sum\limits_{\hat{V}^{\bar{t}}(\tilde{S}_j)\in\mathcal{L}_2}|k\rangle_0|\tilde{S}_k\rangle_1|\hat{V}^{\bar{t}}(\tilde{S}_k)\rangle_2|\bar{V}\rangle_3|1\rangle_4.
    \end{align}
    \item Undo the QFT and QPCA process (as demonstrated in Sec.~\ref{Sec:VI}), 
    %to obtain $ |\Phi^1\rangle$, 
    the quantum state becomes
    \begin{align}
        |\Phi^2\rangle=\sqrt{\frac{\|\mathcal{L}_1\|}{L}}\sum\limits_{\hat{V}^{\bar{t}}(\tilde{S}_j)\in\mathcal{L}_1}|k\rangle_0|\tilde{S}_k\rangle_1|0^m\rangle_2|\bar{V}\rangle_3|0\rangle_4+\sqrt{\frac{\|\mathcal{L}_2\|}{L}}\sum\limits_{\hat{V}^{\bar{t}}(\tilde{S}_j)\in\mathcal{L}_2}|k\rangle_0|\tilde{S}_k\rangle_1|0^m\rangle_2|\bar{V}\rangle_3|1\rangle_4.
    \end{align}
    \item Now trace over the 2nd and 3rd registers, we have $|\Phi^3\rangle\langle\Phi^3|={\rm Tr}_{2,3}[|\Phi^2\rangle\langle\Phi^2|]$, where
    \begin{align}
        |\Phi^3\rangle=\sqrt{\frac{\|\mathcal{L}_1\|}{L}}\sum\limits_{\hat{V}^{\bar{t}}(\tilde{S}_j)\in\mathcal{L}_1}|k\rangle_0|\tilde{S}_k\rangle_1|0\rangle_4+\sqrt{\frac{\|\mathcal{L}_2\|}{L}}\sum\limits_{\hat{V}^{\bar{t}}(\tilde{S}_j)\in\mathcal{L}_2}|k\rangle_0|\tilde{S}_k\rangle_1|1\rangle_4.
    \end{align}
    %\JW{should this Trace be only over 2}.
    %\YW{quantum state $|\psi_2\rangle$ has two registers (say $|k\rangle|S_k\rangle$). To maintain the dimension, subsystems 2,3 should be traced. This seems to be correct.}
    As a result, we can utilize the swap-test combined with quantum phase estimation to 
     provide an $\epsilon$-approximation to
    \begin{eqnarray}
        \begin{split}
            &\|\left(\langle\psi_2|\otimes\langle0|\right)|\Phi^3\rangle\|=\abs{\left(\frac{1}{\sqrt{\sum_{j}(V^{\Bar{t}}(\tilde{S}_j))^2}}\sum_{j\in\{0,1\}^n}V^{\Bar{t}}(S_j)\langle j|\langle\tilde{S}_j|\langle0|\right)|\Phi^3\rangle}\\
        &=\frac{\|\mathcal{L}_1\|^{1/2}}{L}\sum\limits_{\hat{V}^{\bar{t}}(\tilde{S}_k)\in\mathcal{L}_1}\hat{V}^{\bar{t}}(\tilde{S}_k),
        \end{split}
    \end{eqnarray}
    where the quantum state $|\psi_2\rangle$ is given in Eq.~\ref{Eq:LoadGridInformation}, and $\hat{V}^{\bar{t}}(\tilde{S}_k)=V^{\Bar{t}}(S_j)/\sqrt{\sum_{j}(V^{\Bar{t}}(\tilde{S}_j))^2}$. The above quantum state overlap can be efficiently estimated by utilizing a $\mathcal{O}(G/\epsilon)$-depth quantum circuit, where $G$ represents the quantum circuit complexity in preparing $|\psi_2\rangle$ and $|\Phi^3\rangle$.
    \item Recall that the selection of VaR $\bar{V}$ enabling ${\rm Pr}_{j\in[L]}(\hat{V}^{\bar{t}}(\tilde{S}_j)\leq\bar{V})=\|\mathcal{L}_1\|/L=5\%$, as a result, CVaR can be predicted by
    \begin{align} \label{eq:CVaR}
        {\rm CVaR}=\frac{\|\langle\psi_2|\otimes\langle0|\Phi_3\rangle\|}{(5\%)^{3/2}L^{1/2}}=\frac{\|\mathcal{L}_1\|^{1/2}}{(5\%)^{3/2}L^{3/2}}\sum\limits_{\hat{V}^{\bar{t}}(\tilde{S}_k)\in\mathcal{L}_1}\hat{V}^{\bar{t}}(\tilde{S}_k)=\frac{\sum_{\hat{V}^{\bar{t}}(\tilde{S}_j)\in\mathcal{L}_1}\hat{V}^{\bar{t}}(\tilde{S}_j)}{\|\mathcal{L}_1\|}.
    \end{align}
\end{itemize}

\section{Complexity Analysis}
Here, we summarize all involved complexity in Steps~1-4 as demonstrated above. In step~1, the QSVT framework is used to prepare the quantum state $|V^{\Bar{t}}\rangle$. As shown in Fig.~\ref{fig:1}~(a), a constant-depth quantum circuit $U_{\tilde{M}}$ suffices to provide a $(1,3,0)$-block-encoding of the $s$-sparse matrix $\tilde{M}/s$, given oracles $U_R$ and $U_c$. Using $U_{\tilde{M}}$ and the QSVT technique, a $(1,4,0)$-block-encoding to $P(\tilde{M})$ can be constructed, where the involved quantum circuit depth is upper bounded by $\mathcal{O}((T-\bar{t})\Delta\tau^{-1}\|\tilde{M}\|_2\log(1/\epsilon_1))$. The error $\epsilon_1$ represents the upper bound on
$\|\left(\langle10^{\otimes 3}|\otimes I_n\right)U_{\Phi} \left(|10^{\otimes 3}\rangle\otimes I_n\right)|V^T\rangle-|V^{\bar{t}}\rangle\|_2\leq\epsilon_1$. Noting that we essentially utilize the discretisation method to approximate the exact option value function $V$, as a result, $\Delta\tau$-length time slice may introduce $\epsilon_d=(\Delta\tau)^{2}$ additive error. Then the corresponding quantum circuit depth can be approximated by $\mathcal{O}((T-\bar{t})\epsilon_d^{-1/2}\|\tilde{M}\|_2\log(1/\epsilon_1))$.

In step~2, a quantum circuit implementation of ``classical MC" is implemented, which is independent of the result given in step~1. Here, a $\mathcal{O}(\bar{t}\Delta\tau^{-1})$-depth quantum circuit is designed to prepare the quantum state $|\phi^{\bar{t}}\rangle$. Each layer utilizes quantum control gates to compute functions $\mathcal{F}(\cdot)$ and
$\mathcal{F}^{-1}(\cdot)$, whose quantum gate complexity can be estimated by $\mathcal{O}(\log(L){\rm poly}(\log(1/\epsilon_2)))$ by using the Fourier-Transformation-based method~\cite{zhou2017quantum}, where $L$ represents the number of stock prices and $\epsilon_2$ represents the additive error in approximating functions $\mathcal{F}(\cdot)$ and
$\mathcal{F}^{-1}(\cdot)$.

Step~3 utilized the QPCA and QPE methods to combine option values (given by $|V^{\bar{t}}\rangle$) with the stock prices state $|\phi^{\bar{t}}\rangle$. To simulate $e^{-i\rho\tau}$, where the density matrix $\rho$ is provided by Eq.~\ref{Eq:density_matrix}, we first divide the evolution time $\tau$ into $N$ time-slice $\Delta t$, then utilize ${\rm Tr}_{A}[e^{-i\omega\Delta t}\rho_A\otimes\sigma_B e^{i\omega\Delta t}]$ to approximate $e^{-i\rho\Delta t}$ with $\mathcal{O}(\Delta t^2)$ additive error. Therefore, the total error in simulating $e^{-i\rho\tau}$ can be approximated by $\epsilon_3\leq N(\Delta t)^2=N(\tau/N)^2$, resulting in the sample complexity $N\leq\mathcal{O}(\epsilon_3^{-1}\tau^2)$. Furthermore, to encode the spectrum information of $\rho$ into the quantum state $|\phi^{\bar{t}}\rangle$, the QPE algorithm is required by utilizing the operator, as shown in Eq.~\ref{Eq:qpe}. To provide an estimation within $\epsilon$ error, the evolution time (in Eq.~\ref{Eq:qpe}) should satisfy $\left(N_{\rm qpe}\Delta t\right)=\tau$. As a result, the sample complexity on $|V^{\Bar{t}}\rangle$ in performing quantum phase estimation meanwhile generating $|\Phi\rangle$ is 
\begin{align}
    \mathcal{O}\left(\epsilon_3^{-1}\left(N_{\rm qpe}\Delta t\right)^{2}\right)=\mathcal{O}\left(\epsilon_3^{-3}\right).
    \label{Eq:QPCAComplexity}
\end{align}
The number of qubits is upper bounded by $\mathcal{O}(n+\log(1/\epsilon_3))$.

However, step~1 does not output the exact $|V^{\bar{t}}\rangle$ without any error. Actually, it outputs an $\epsilon_1$-approximation to $|V^{\bar{t}}\rangle$. Without loss of generality, we assume step~3 essentially performs QPCA on 
\begin{align}
    |\tilde{V}^{\bar{t}}\rangle=\frac{1}{\sqrt{1+\epsilon_1^2}}|V^{\bar{t}}\rangle+\frac{\epsilon_1}{\sqrt{1+\epsilon_1^2}}|e\rangle,
\end{align}
where $|e\rangle$ represents the error component. It is easy to verify that $\||\tilde{V}^{\bar{t}}\rangle-|V^{\bar{t}}\rangle\|_2\leq\epsilon_1$. Performing the gird computation gate (as shown in Eq.~\ref{Eq:LoadGridInformation}), the quantum system occupies 
\begin{align}
    |\tilde{\phi}^{\bar{t}}\rangle=\frac{1}{\sqrt{(1+\epsilon_1^2)\sum_{j}(V^{\Bar{t}}(\tilde{S}_j))^2}}\sum_{j\in\{0,1\}^n}V^{\Bar{t}}(S_j)|j\rangle_1|\tilde{S}_j\rangle_2+\frac{\epsilon_1}{\sqrt{1+\epsilon_1^2}}{\text{C}}U_S\left(|e\rangle_1|0^m\rangle_2\right).
\end{align}
Let the error component $|e\rangle=\sum_je_j|j\rangle$, we have
\begin{eqnarray}
\begin{split}
     {\rm Tr}_1\left[|\tilde{\phi}^{\bar{t}}\rangle\langle\tilde{\phi}^{\bar{t}}|\right]&=\frac{1}{(1+\epsilon_1^2)\sum_{j}(V^{\Bar{t}}(\tilde{S}_j))^2}\sum\limits_{j}V^{\bar{t}}(\tilde{S}_j)^2|\tilde{S}_j\rangle\langle \tilde{S}_j|+\frac{\epsilon_1^2}{1+\epsilon_1^2}\sum_ke_k^2|\tilde{S}_k\rangle\langle\tilde{S}_k|\\
     &+\frac{\epsilon_1}{(1+\epsilon_1^2)}{\rm Tr}_1\left[|V^{\bar{t}}\rangle_{12}\langle\psi|_{12}\right]+\frac{\epsilon_1}{(1+\epsilon_1^2)}{\rm Tr}_1\left[|\psi\rangle_{12}\langle V^{\bar{t}}|_{12}\right],
\end{split}
\end{eqnarray}
where $|\psi\rangle={\text{C}}U_S\left(|e\rangle_1|0^m\rangle_2\right)$. This naturally results in
\begin{align}
    \left\| {\rm Tr}_1\left[|\tilde{\phi}^{\bar{t}}\rangle\langle\tilde{\phi}^{\bar{t}}|\right]- {\rm Tr}_1\left[|\phi^{\bar{t}}\rangle\langle\phi^{\bar{t}}|\right]\right\|_2\leq \frac{\epsilon_1^2}{1+\epsilon_1^2}\left\|{\rm Tr}_1\left[|\phi^{\bar{t}}\rangle\langle\phi^{\bar{t}}|\right]\right\|_2+\frac{\epsilon_1^2}{1+\epsilon_1^2}+\frac{2\epsilon_1}{1+\epsilon_1^2}\leq 4\epsilon_1.
\end{align}
The demonstrated $L_2$ norm difference upper bound implies all eigenvalues of ${\rm Tr}_1\left[|\phi^{\bar{t}}\rangle\langle\phi^{\bar{t}}|\right]$ can be approximated by that of ${\rm Tr}_1\left[|\tilde{\phi}^{\bar{t}}\rangle\langle\tilde{\phi}^{\bar{t}}|\right]$ with at most $4\epsilon_1$ additive error. Combining this error analysis with the sample complexity given in Eq.~\ref{Eq:QPCAComplexity}, let $\epsilon_1=\epsilon/8$ and $\epsilon_3=\epsilon/2$, then repeat step~1 $\mathcal{O}(8\epsilon^{-3})$ times. This suffices to prepare an $\epsilon$-approximation to $|\Phi\rangle$.

Finally, step~4 utilized the binary search program and mean value estimation algorithm, where each iteration step requires $\mathcal{O}(\epsilon_4^{-1})$ copies of $|\Phi\rangle$ to estimate VaR and CVaR. Combining all steps together, the quantum circuit depth can be approximated by 
\begin{align}
    \mathcal{O}\left(\frac{\bar{t}\epsilon_d^{-1/2}\log(L){\rm poly}(\log(1/\epsilon))}{\epsilon}+\frac{8(T-\bar{t})\epsilon_d^{-1/2}\|\tilde{M}\|_2\log(8/\epsilon)}{\epsilon^4}\right).
\end{align}
We summarize the required quantum computational resources for each step in Table~\ref{tab:Compare}.

\begin{table}[h]
\caption{Quantum Computational Resources Summary}
 \label{tab:Compare}
\begin{center}
\begin{tabular}{c|cccc}
    \hline\hline
   \textbf{Alg. Steps} & \textbf{Assumption} & \textbf{Circuit Depth} & \textbf{Post-Selection} & \textbf{Sample Complexity}\\
    \hline
    \multirow{1}{*}{Step~1} & Efficient access to $\tilde{M}$ & $\mathcal{O}(T-\bar{t})\Delta\tau^{-1}\|\tilde{M}\|_2\log(8/\epsilon))$ & $\mathcal{O}(1)$ & $\mathcal{O}(1)$\\
    \hline
    \multirow{1}{*}{Step~2} & --- & $\mathcal{O}(\bar{t}\Delta\tau^{-1}\log(L){\rm poly}(\log(1/\epsilon)))$ & --- & $\mathcal{O}(1)$\\
    \hline
    \multirow{1}{*}{Step~3} & --- & $\mathcal{O}(\log(1/\epsilon))$ & --- & $\mathcal{O}(8\epsilon^{-3})$-copies of $|V^{\bar{t}}\rangle$~(Eq.~\ref{Eq:Vt})\\
    \hline
    \multirow{1}{*}{Step~4} & --- & $\mathcal{O}(1/\epsilon)$ & --- & $\mathcal{O}(\log(1/\epsilon))$\\
    \hline\hline
    \end{tabular}
\end{center}
\end{table}

\section{American Option Values}
%\subsection{Computing American Option Value}
For American option pricing, finite difference approximation of the partial differential equation complementarity problem \eqref{compl} leads to a linear complementarity problem.
While more sophisticated iterative methods, e.g., penalty method \cite{ForsythPenalty02} and Newton method \cite{CLV02} can be applied to compute the solution of this linear complementarity problem, here we attempt to implement a rudimentary 
%simple
method that explicitly handles the inequality payoff constraint, as described in Alg.~\ref{AmericanOptionValue}. Note that, if explicit finite difference approximation is used, this 
simple scheme does solve the resulting linear complementarity finite difference approximation of Eq.~\ref{compl}.

\begin{algorithm}
\label{AmericanOptionValue}
\caption{American Option Value Computation}
\textbf{Input:} $\vec{V}^T=(V_0^T,V_1^T,\cdots, V_{2^n-1}^T)=(\payoff(S_0),\payoff(S_1),\cdots, \payoff(S_{2^n-1}))$, matrix $M$, $\bar{t}$.\\
\textbf{Output:} $\vec{V}^0$\\
{\textbf{for}} $t=T, T-\Delta\tau\ldots, \bar{t}+\Delta\tau, \bar{t}$\\
   \quad \quad \textbf{solve} ~ $ \hat{V}^{t-\Delta\tau} = [I+M]^{-1}V^{t}$.\\
    \quad \quad \textbf{for}~$j=0,\ldots,2^n-1$\\
   \quad \quad  \quad \quad $V_j^{t-\Delta\tau} =  \max(\payoff(S_j), \hat{V}_j^{t-\Delta\tau})$\\
  \quad \quad {\textbf{End for}}\\
    {\textbf{End for}}\\
\end{algorithm}

As shown in Alg.~\ref{AmericanOptionValue}, 
the option value $V^{t}$ is updated by two steps: (i) solving a linear system $(I+M)\hat{V}^{t-\Delta\tau}=V^{t}$ and (2) taking the higher value of $\hat{V}_j^{t-\Delta\tau}$ and $\payoff(S_j)$. Now encode the option values in the amplitudes in the same way as the European option values given in Eq.~\ref{Eq:Vt}, and the representative quantum state is then $|V^{t}\rangle=\sum_jV_j^{t}|j\rangle$.  Given an efficient oracle to access entries within the matrix $M$, there are several efficient quantum algorithms for preparing the quantum state $|\hat{V}^{t-\Delta\tau}\rangle=[I+M]^{-1}|V^{t}\rangle$. However, we argue that achieving the comparison between  $\hat{V}_j^{t-\Delta\tau}$ and $\payoff(S_j)$ is quantum hard. 

\begin{definition}[Amplitude Maximum~(AM) Problem]
    Given $N$ copies of valid quantum states $|x\rangle=\sum_{i=0}^{2^n-1}x_i|i\rangle$ and $|y\rangle=\sum_{i=0}^{2^n-1}y_i|i\rangle$ with $n$ qubits, where $x_i,y_i\in\mathbb{R}$, prepare the quantum state 
    \begin{align}
        |z\rangle=\frac{\sum\limits_{i=0}^{2^n-1}\max(x_i,y_i)|i\rangle}{\sqrt{\sum\limits_{i=0}^{2^n-1}\max^2(x_i,y_i)}}
        \label{Eq:max_z}
    \end{align}
    by using a quantum computer.
\end{definition}

In what follows, we demonstrate the difficulty of solving the American option problem using the above proposed algorithm with quantum computers. In other words, the method described above would necessitate the same $\mathcal{O}(2^n)$ operations as classical computing, i.e. without any quantum advantage. 
%-- what do you think?

%\YL{Jingbo: delete the word "even", since this is not an issue in classical computing?} 

%\YW{Yuying, given the AM problem defined above, no polynomial classical algorithm can solve this problem in the worst-case scenario, otherwise it may imply complexity class BQP living in BPP which is believed implausible. An intuition is: classical algorithm should compare each $x_i, y_i$ for all $i\in[2^n]$.}

\begin{theorem}
    Given $N$ copies of valid quantum states $|x\rangle=\sum_{i=0}^{2^n-1}x_i|i\rangle$ and $|y\rangle=\sum_{i=0}^{2^n-1}y_i|i\rangle$ with $n$ qubits and real amplitudes $x_i,y_i\in\mathbb{R}$, any quantum algorithm that generates $|z\rangle$ (as shown in Eq.~\ref{Eq:max_z}) will require $N\geq\Omega(2^n)$ samples.
\end{theorem}
\begin{proof}
    We first consider the sample complexity of a quantum state classification problem: how many samples suffice to distinguish quantum states $|\psi\rangle=-\sqrt{\frac{d-1}{d}}|0\rangle^n+\sqrt{\frac{1}{d}}|1\rangle^n$ and $|\phi\rangle=|0\rangle^n$, where $d=2^n$. It is known that quantum states $\rho$ and $\sigma$ are distinguishable if $\|\rho-\sigma\|_1\geq0.8$~\cite{cotler2021revisiting}. Now suppose we have $m$ copies of $|\psi\rangle$ and $|\phi\rangle$, then
    \begin{align}
        \|(|\psi\rangle\langle\psi|)^{\otimes m}-(|\phi\rangle\langle\phi|)^{\otimes m}\|_1\geq 0.8
    \end{align}
    may directly result in $\sqrt{1-(1-1/d)^m}\geq0.8$, equivalently $m\geq\Omega(2^n)$.

    Now suppose there exists a quantum algorithm $\mathcal{A}$ that can solve the AM problem with $N<\mathcal{O}(2^n)$ copies of $|x\rangle$ and $|y\rangle$, that is $\mathcal{A}((|x\rangle\langle x|)^{\otimes N},(|y\rangle\langle y|)^{\otimes N})=|z\rangle\langle z|$. Let the testing quantum state $|{\rm test}\rangle=|1\rangle^n$, and consider a classification program 
    \begin{align}
        \mathcal{A}(\cdot, |{\rm test}\rangle\langle{\rm test}|).
    \end{align}
 This leads to 
    \begin{itemize}
        \item $\mathcal{A}(|\psi\rangle\langle\psi|^{\otimes N}, |{\rm test}\rangle\langle{\rm test}|^{\otimes N})=(0,0,\cdots,0,1)=|1\rangle^n$
        \item $\mathcal{A}(|\phi\rangle\langle\phi|^{\otimes N}, |{\rm test}\rangle\langle{\rm test}|^{\otimes N})=\left(\frac{1}{\sqrt{2}},0,\cdots,0,\frac{1}{\sqrt{2}}\right)=\frac{1}{\sqrt{2}}\left(|0\rangle^n+|1\rangle^n\right)$ 
    \end{itemize}
    which imply $|\psi\rangle$ and $|\phi\rangle$ could be distinguished by a quantum algorithm with $N<2^n$ copies and constant number of computational basis measurements. This results in a contradiction.
\end{proof}

\section{Conclusion}
% In the context of finance, computational time and accuracy may directly affect the profit and loss of the business for which the problems are being solved. In other words, any actual slightly speedup and associated model performance improvement may have a tremendous impact on the financial industry. For example,
% fast and accurate evaluation of the risk metrics in derivatives trading is crucial in effectively hedging the risks especially under volatile market conditions.
%~\cite{herman2023quantum}. 
%Although existing quantum algorithms have been proposed in accelerating Monte Carlo methods, it is still unclear whether existed approaches for stochastic modelling, optimization and quantum Monte Carlo could be turned into an ``end-to-end'' quantum advantages on practical problems.

In finance, even modest improvements in computational speed and model performance can have a substantial effect on the profitability of a business. For instance, fast and accurate evaluation of the risk metrics in derivatives trading is crucial in effectively hedging the risks especially under volatile market conditions.

This paper introduces an efficient end-to-end quantum algorithm for predicting the VaR/CVaR of a portfolio of European options. Specifically, given the backward propagation time target $\bar{t}$ and accuracy $\epsilon$, our quantum algorithm takes market sensitive parameters as inputs and generates a $\epsilon$-approximate VaR/CVaR by running a quantum algorithm with $\tilde{\mathcal{O}}(\max\{\bar{t},T-\bar{t}\}\epsilon_d^{-1/2}\epsilon^{-4})$ time complexity, where $\epsilon_d$ represents the error induced by the discretization approach. The essential quantum speed-up relies on matching the stock price $\tilde{S}_j^{\bar{t}}$ to its corresponding option value $\hat{V}(\tilde{S}_j^{\bar{t}})$, where the option value function lives in a high-dimensional space induced by discretization. In general, classical approaches would require transverse all option values in the look-up table, however, QPCA and QPE provide an efficient approach to project all concerned option value functions to the concerned stock prices.

%This work leaves room for further research. For example, we have only considered the single stock option price in this paper. Multi-stocks naturally induce a high-dimensional fair option value PDE. Complex correlations between different stocks within PDE poses more challenges to classical algorithms, as a result, how to extend our quantum algorithm to multiple underlying stocks would be important in practice. Furthermore, whether different discretization strategies of the American option PDE can bypass the proved `no-go' theorem deserves to be analyzed.

This work opens new avenues for further research. For instance, we have only considered the single stock option pricing in this paper to demonstrate potential quantum advantages of the proposed algorithm. Multi-stock scenarios naturally induce a high-dimensional fair option value PDE. The curse of dimensionality and the complex correlations between different stocks within this PDE pose formidable challenges to classical algorithms. Therefore, extending this quantum algorithm to multiple underlying stocks is an obvious next step. Furthermore, analyzing whether different discretization strategies for the American option PDE can bypass the established 'no-go' theorem is a topic deserving further investigation.

\clearpage
\bibliography{ref}

\section*{Acknowledgements}
The authors wish to express gratitude to Peter Forsyth, Song Wang, Robert Scriba, Tavis Bennett, Sisi Zhou, Des Hill for valuable discussions. Y.~Wu is supported by the China Scholarship Council (Grant No.~202006470011) and UWA-CSC HDR Top-Up Scholarship (Grant No.~230808000954).

\end{document}